\documentclass[12pt]{amsart}
\usepackage{amssymb,amscd,latexsym,amsmath,graphics}
\usepackage[all]{xy}
\usepackage{mathrsfs}
\usepackage[T1]{fontenc}
\usepackage[utf8]{inputenc}
\usepackage[english]{babel}
\usepackage{lmodern}
\usepackage{url}
\usepackage{geometry}
\usepackage{blkarray}
\usepackage{xcolor}
\usepackage{mathabx}
\usepackage{lipsum}

\newcommand\blfootnote[1]{%
  \begingroup
  \renewcommand\thefootnote{}\footnote{#1}%
  \addtocounter{footnote}{-1}%
  \endgroup
}

\newtheorem{teo}{Theorem}[section]

\newtheorem{lem}[teo]{Lemma}

\newtheorem{defi}[teo]{Definition}
\newtheorem{ex}[teo]{Example}

\newtheorem{prop}[teo]{Proposition}
\newtheorem{rem}[teo]{Remark}

\newtheorem*{conjecture}{Conjecture}

\newcommand{\F}{{\mathbb F}}

\makeindex

\title[On a group under which symmetric reed-muller codes are invariant]{On a group under which symmetric reed-muller codes are invariant}

\author{Sibel Kurt Toplu$^*$}

\address{Hacettepe University, Graduate School of Science and Engineering, Beytepe, Ankara, Türkiye}
\email{sibel.toplu@tubitak.gov.tr}
\author{Talha Arıkan}
\address{Hacettepe University, Department of Mathematics,Ankara, Türkiye}
\email{tarikan@hacettepe.edu.tr}
\author{Pınar Aydoğdu}
\address{Hacettepe University, Department of Mathematics,Ankara, Türkiye}
\email{paydogdu@hacettepe.edu.tr}

\author{Oğuz Yayla}
\address{Middle East Technical University, Institute of Applied Mathematics, 06800, Ankara, Türkiye}
\email{oguz@metu.edu.tr}

\begin{document}
\begin{abstract}
The Reed-Muller codes are a family of error-correcting codes that have been widely studied in coding theory. In 2020, Wei Yan and Sian-Jheng Lin introduced a variant of Reed-Muller codes so called symmetric Reed-Muller codes. We investigate linear maps of the automorphism group of symmetric Reed-Muller codes and show that the set of these linear maps forms a subgroup of the general linear group, which is the automorphism group of punctured Reed-Muller codes. 
We provide a method to determine all the automorphisms in this subgroup explicitly for some special cases. 


\medskip


\medskip


\noindent \textbf{MSC2020-Mathematics Subject Classification:}  94B05, 11T71, 08A35, 11T06 
\end{abstract}
\maketitle
\section{ Introduction and Preliminaries }

\label{sec:intro}

\blfootnote{$^*$ Corresponding Author}\blfootnote{This publication is a part of the doctoral thesis of Sibel Kurt Toplu} 

Working with automorphisms of the codes is important for various reasons in the field of coding theory. The study of automorphisms can contribute to developing efficient decoding algorithms. By exploiting the symmetries within a code, it may be possible to design more computationally efficient algorithms in order to correct errors.
Thus automorphisms of the codes play a crucial role in understanding and characterizing the error-correction properties of codes. To sum up by studying the symmetries of a code, one can gain insights into how errors affect the encoded information and how the code can be designed to correct or detect these errors.

Reed-Muller codes ($RM$ codes, for short) are a family of error-correcting codes that were first introduced by Irving S. Reed and Gustave Solomon Muller in 1954. A large number of $RM$ code variations and generalizations were introduced in the literature, for instance, see \cite{delsarte1970generalized,mac,pless1998handbook}. 
They have a large minimum distance, which makes them good at correcting errors. They also have simple encoding and decoding algorithms, which make them useful and efficient to implement. There are some variants of $RM$ codes, such as the binary $RM$ codes defined on the prime field $\F_2$ and the $p$-ary $RM$ codes defined on the prime field $\F_p$, where $p$ is a prime number. 
Furthermore, $RM$ codes have been generalized in many ways, such as the generalized Reed-Muller codes ($GRM$, for short), which are defined over an arbitrary finite field, and the symmetric Reed-Muller codes ($SRM$, for short) which have a certain symmetry property.


In coding theory, the majority of classical codes have a sizable automorphism group that is connected to a linear group. For example, $GRM$ codes are invariant under the general affine group $GA(m,q)$ \cite{delsarte1970generalized}. A requirement for a code to be invariant under the affine group is provided by Kasami et al. in \cite{kasami}. Delsarte describes the codes as invariant under certain linear groups in \cite{delsartecyclic}.  He shows that under $GA(m,q)$, only $p$-ary $RM$ codes can be invariant; nonetheless, the issue of fully determining the automorphism group of affine invariant codes has not yet been resolved.
Determining the complete automorphism group of a code is a challenging topic that is frequently connected to the categorization of simple groups. 
Dür is credited with the research of the automorphism groups of Reed-Solomon codes and their extensions \cite{delsartecyclic}. In \cite{berger1993automorphism}, Berger has proved the complete automorphism groups of $GRM$ codes. In 1996, Berger has demonstrated that $GA(m,q)$ contains the permutation group of any affine-invariant code \cite{berger1996automorphism} and then he demonstrates how to create a formal expression for every affine-invariant code's permutation group in \cite{berger1996permutation}.
The complete automorphism groups of the projective and homogeneous $RM$ codes are found in \cite{berger2006automorphism}. These groups are associated with the projective linear group and the general linear group, respectively. In this work, we aim to investigate a subgroup of the automorphism group $SRM$ codes whose elements are linear maps.

$SRM$ codes are first introduced by using bivariate polynomials over $\mathbb{F}_q$ in \cite{yan2020symmetric}. The local correctability of the bivariate $SRM$ codes is discussed in \cite{yan2021local}. The authors begin by outlining the advantages of the symmetric structure and offering intuitions to indicate the superiority of $SRM$ codes over $GRM$ codes in terms of local correctability. The tolerance of error ratios of $SRM$ codes and $GRM$ codes over the same finite field, as well as their code rate, are then demonstrated.
They establish that a class of locally-correctable codes is composed of multivariable $SRM$ codes $SRM_q[n,r]$ in \cite{yan2021local}. Furthermore, the dual of $SRM_q[n,r]$ is also presented.
However, transformations preserving these codes, specifically belonging to the automorphism group, have not been studied yet. Taking inspiration from this gap in this research, we have directed our focus toward addressing this problem.

First, we will recall some basic notions and some results on $RM$ and $GRM$ codes which will be useful throughout the paper.

\subsection{Equivalence of Linear Codes} Let $\F_q$ be the finite field of order $q=p^m$ for a prime $p$ and $\F_q^n$ be the $n$-dimensional vector space over $\F_q$. The measure of dissimilarity between two vectors is established by counting the coordinates in which they differ. A linear $[n,k,d]$-code is $k$-dimensional linear subspace of $\F_q^n$ with the minimum distance $d$. A generator matrix $G$ of a linear $[n,k]$-code $C$ is any matrix of row rank $k$, whose rows come from the code $C$. We will give the definitions of some types of code-equivalence. We refer to \cite{HP} for further details. We need the following definition for simplicity. 

\begin{defi}\cite{brualdi2020permutation}
Let $\pi = (\pi_1, \pi_2, \ldots, \pi_n)$ be a permutation of $\{1, 2, \ldots , n\}$. The permutation $\pi$ can be given in the equivalent form as an $n \times n$ permutation
matrix $P$ with $1$’s in positions $(i, \pi_i)$ for $ i = 1, 2, \ldots , n$ and $0$’s elsewhere. $\mathcal{P}_n$ denotes the corresponding
set of all $n \times n$ permutation matrices $P$.
\end{defi}

Two linear codes $C_1$ and $C_2$ are called {\em permutation equivalent} if there is a permutation of coordinates which sends $C_1$ to $C_2$. Hence, two linear codes $C_1$ and $C_2$ of the same length are permutation equivalent if there exists a permutation matrix $P \in \mathcal{P}_n$  such that 
\[
G_2=G_1P,
\]
where $G_1$ and $G_2$ are the generator matrices of the codes $C_1$ and $C_2$, respectively. 

If we work on a finite field except $\mathbb{F}_2$, then we may need a more general form of the equivalence. Recall that a square matrix is called {\em monomial} if it has exactly one nonzero entry from that finite field in each column and row. Hence, every monomial matrix over $\mathbb{F}_2$ is a permutation matrix. A monomial matrix $M$ can be expressed as either $DP$ or $PD_1$, where $D$ and $D_1$ are non-singular diagonal matrices over the corresponding finite field and $P$ is a permutation matrix. 
Two linear codes $C_1$ and $C_2$ of the same length over $\F_q$ are said to be {\em monomially equivalent} when there exists a monomial matrix $M$ such that 
\[
G_2=G_1M,
\]
where $G_1$ and $G_2$ are the generator matrices of the codes $C_1$ and $C_2$, respectively. Note that monomial equivalence and permutation equivalence coincide for binary codes.

Let $\gamma$ be a field automorphism of $ \mathbb{F}_q$ and $M=DP$ be a monomial matrix over $ \mathbb{F}_q$, where $P$ is a permutation matrix and $D$ is a non-singular diagonal matrix over $\mathbb{F}_q$. Applying the map $M\gamma$ to any codeword is described by the following process: Firstly, the $i^{th}$ component of code is multiplied by the $i^{th}$ diagonal entry of $D$ for all $i$'s. Then the corresponding permutation associated with the permutation matrix $P$ is applied to codeword. Finally, the automorphism $\gamma$ is applied to all components.
Two linear codes $C_1$ and $C_2$ of the same length over $\F_q$ are said to be {\em equivalent} when there is a monomial matrix $M$ and a field automorphism $\gamma$ of $F_q$ such that
\[
C_2=C_1M\gamma,
\]
where $C_1M\gamma$ is obtained by applying $M\gamma$ to each codeword of  $C_1$. This is the most general form of the equivalence.

Note that, all equivalence definitions are the same for the binary codes. Furthermore, monomial equivalence and general equivalence coincide for $p$-ary codes, where $p$ is a prime. 

Since we have three types of equivalence, there exist three possible definitions of the automorphism groups of the code families by considering $C_1=C_2$ in the above definitions.

Now consider a code $C$ of length $n$ over the field $\F_q$. The set of coordinate permutations that map the code $C$ to itself forms a group, called the {\em permutation automorphism group of $C$} and denoted by $PAut(C)$. Obviously, $PAut(C)$ is a subgroup of the symmetric group $S_n$. The set of monomial matrices, by which $C$ is monomially equivalent to itself, forms the group $MAut(C)$, which is called the {\em monomial automorphism group of $C$}. The set of maps of the form $M\gamma$, where $M$ is a monomial matrix and $\gamma$ is a field automorphism, that map $C$ to itself forms the group $Aut(C)$, called {\em automorphism group of $C$}. 

In general, we have that $PAut(C)\subseteq MAut(C)\subseteq Aut(C)$.  If $q = 2$, then $PAut(C) = MAut(C) = Aut(C)$. If $q$ is prime, then $MAut(C) =Aut(C)$.

All these definitions and conclusions are well-known in the literature and can be found in any basic coding theory book, for example in  \cite{HP}.

\subsection{The Automorphisms of Reed-Muller Codes}


There exist some equivalent definitions of the $RM$ codes. Following \cite{mac}, $RM$ codes can be defined in terms of multivariable polynomials as follows: Let $x = (x_1, \ldots, x_m)$ range over ${\mathbb{F}_2}^m$. Any function $f(x) = f(x_1, \ldots, x_m)$ which takes the values $0$ and $1$ is called {\em a binary multivariable function}. We recall the following definitions.


\begin{defi}\cite{mac}
The {\em $r^{th}$ order binary $RM$ code $R(r,m)$ of length $n=2^m$} for $0\leq r \leq m$, is the set of all vectors $f$, where $f(x_1, \ldots, x_m)$ is a binary multivariable polynomial of degree at most $r$, i.e.,
$$R(r,m)= \bigl\{(f(\alpha))_{\alpha \in \mathbb{F}_2^m} \mid f \in \mathbb{F}_2[x_1,\ldots, x_m],\,\, deg(f) \leq r \bigr\}.$$
\end{defi}

\begin{defi}\cite{mac}
For $0 \leq r \leq m -1$, the $RM$ code which is obtained by puncturing (or deleting) the coordinate corresponding to $ x_1 = \cdots = x_m = 0 $ from all the codewords of $R(r,m)$ is called {\em the punctured $RM$ code}, and it is denoted by $R(r,m)^*$.
\end{defi}

Now we will mention the notion of affine invariance which is crucial for our discussion on automorphism groups. Before going further, we need some basic notions.

Let $A=[a_{ij}]$ be an invertible $m\times m$ binary matrix and $\textbf{b}$ be a binary $m\times 1$ vector. Consider the transformation $T$ from binary $m$-tuples to binary $m$-tuples defined by  
\[
T :
  \begin{bmatrix}
    x_1 \\
    x_2 \\
    \vdots \\
    x_m
  \end{bmatrix}
  \mapsto A
  \begin{bmatrix}
    x_1 \\
    x_2 \\
    \vdots \\
    x_m
  \end{bmatrix} + \textbf{b},
\]
which permutes binary $m$-tuples. $T$ can be also considered as a permutation of multivariate polynomials as follows:

\begin{align}
    \label{eqn:trans}
    T_f(A,\textbf{b}) : f(x_1, \ldots, x_m) \mapsto f\Bigl(\sum a_{1j}x_1 + b_1, \ldots, \sum a_{mj}x_j + b_j\Bigl).
\end{align}
The set of all such transformations formed by $T$ is a group, which is known as the \textit{general affine group} over $\mathbb{F}_2$ and is denoted by $GA(m,2)$ (see \cite{mac}). It is obvious that if $f$ is a polynomial of degree $r$, so is $T_f(A,\textbf{b})$.

As we mentioned above, for a binary code it is known that $PAut(C)=Aut(C)$, and hence it is a subgroup of the symmetric group $S_n$. Following \cite{abbe}, a code $C$ is said to be \textit{affine invariant} if $Aut(C)$ includes a subgroup that is isomorphic to the affine linear group.

The following example demonstrates the $RM$ code families are one of the examples of affine invariant codes. This example is important to see the equivalence between transformations applied to variables of the function to evaluate codeword and transformations applied to the codeword itself.

\begin{ex} \cite{abbe}\label{Abbe}    
$R(r,m)$ codes are affine invariant.
\end{ex}

\begin{proof}
Let $A$ be an $m \times m$ invertible matrix over $\mathbb{F}_2$ and $\textbf{b} \in \mathbb{F}_2^{m}$.
The affine linear transform $T : x \mapsto Ax + \textbf{b}$ yields a permutation on the coordinates of the codeword since the codewords of $RM$ codes are evaluation vectors and are indexed by the vectors $x \in \mathbb{F}_2^{m}$. Then such a permutation belongs to $Aut(R(r, m))$. Let $c$ be a codeword in $R(r,m)$. Then there exists a polynomial $f \in \mathbb{F}_2[x_1, \ldots, x_n]$ with $deg(f)\leq r$ such that $c=(f(\alpha))_{\alpha \in \mathbb{F}_{2^m}}$. Since the outcome of the transformation $(f\circ T)(x)$ is another polynomial of degree less than or equal to $r$, we have $c'=(f\circ T)(\alpha)_{\alpha \in \mathbb{F}_2^m}\in R(r,m)$. Thus, $R(r,m)$ codes are affine invariant.
\end{proof}

 Thus, the general affine group $GA(m,2)$ permutes the codewords of the $r^{th}$ order $R(r,m)$ and
$ GA(m,2) \subset \text{Aut}R(r, m)$ (see \cite{mac}).
The subgroup of $GA(m,2)$ consisting of all transformations 

\[
T :
  \begin{bmatrix}
    x_1 \\
    x_2 \\
    \vdots \\
    x_m
  \end{bmatrix}
  \mapsto  A
  \begin{bmatrix}
    x_1 \\
    x_2 \\
    \vdots \\
    x_m
  \end{bmatrix} 
\]
(i.e., for which $\textbf{b}=0$) is known as the \textit{general linear group} and is denoted by $GL(m,2)$. We can consider the transformation $T$ described above in the following way, too:
\begin{equation}\label{def : T(A)}
    T(A):(x_1, \ldots, x_m) \mapsto \Bigl(\sum a_{1j}x_1, \ldots, \sum a_{mj}x_j \Bigl).
\end{equation}
For the sake of convenience in usage, the function $T_f(A,b)$ will be denoted as $T_f(A)$ when $\textbf{b}=0$.
Since the transformation $T(A)$ in (\ref{def : T(A)}) fixes the zero $m$-tuple, the group $GL(m,2)$ permutes the codewords of the punctured $RM$ code ${R}(r,m)^\ast$, i.e., $GL(m,2) \subset \text{Aut} R(r,m)^\ast$ (see \cite{mac}).

We know from \cite[p.~400]{mac} that

\begin{align*}
    Aut(R(r,m)^{*}) &= S_{2^m - 1} \quad \text{for}\quad r = 0 \quad \text{and} \quad m-1, \\
    Aut(R(r,m)) &= S_{2^m} \quad \text{for}\quad r = 0 \quad \text{and} \quad m.
\end{align*}
Furthermore, the following result is also given in \cite[p. 400]{mac} which determines the exact automorphism group of $RM$ codes of $r^{th}$ order with length $2^m$  over $\mathbb{F}_2$ completely.
\begin{teo}\cite[p.~400]{mac}\label{Aut:RM} For $1\leq r \leq m-2$:
\begin{itemize}
    \item  $Aut(R(r,m)^\ast)= GL(m, 2)$,
    \item  $Aut(R(r,m))= GA(m,2)$.
    
\end{itemize}

\end{teo}
    
Note that one may easily adopt the definitions of the general affine group and general linear group on $\mathbb{F}_q$ for any prime $q$. These groups are denoted by $GA(m,q)$ and $GL(m, q)$, respectively.

\subsection{Generalized Reed-Muller Codes}
\label{sec:GRM}

$GRM$ codes are a generalization of $RM$ codes. They are obtained by constructing the codes over any finite field $\mathbb{F}_q$, where $q$ is a prime power. The following is a formal definition of $GRM$ codes:
\begin{defi}\cite{delsarte1970generalized} Let $m$ and $r$ be positive integers. The $GRM$ code of order $r$ with block length $q^m$ over $\mathbb{F}_q$ is defined by
$$GRM_q(m, r])= \bigl\{(f(\alpha))_{\alpha \in \mathbb{F}_q^m} \mid f \in \mathbb{F}_q[x_1,\ldots, x_m],\,\, deg(f) \leq r \bigr\}.$$
\end{defi} 

The $GRM$ codes have the following property:

\begin{teo}\cite{delsarte1970generalized}
    For $0 \leq r
    \leq m(q - 1)$, the automorphism group of $GRM_{q}(m,r)$ codes contains the general affine group $GA(m,q)$ under the natural action on $V = \F_{q^m}$.
\end{teo}

Knorr and Willems in \cite{knorr} give a complete description of the automorphism group of the $p$-ary $RM$ codes for any prime $p$. The automorphism group of the $p$-ary $RM$-codes equals the general affine group $GA(m,p)$.

 In \cite{berger1993automorphism}, Berger and Charpin provide a complete description of the automorphism group of a $GRM$ code. They show that the automorphism group of $GRM$ codes is the affine linear group, i.e., 
 \[
 Aut(GRM_q(m,r)) = GA(m,q).
 \]

\subsection{Outline of the Paper}


Unless otherwise stated, throughout the article we will work on the field $\mathbb{F}_q$, where $q$ is a prime.  Main part of the presented work consists of 
demonstrating the set of transformations under the general linear group that leaves the symmetric Reed-Muller ($SRM$) codes invariant.

Transformations that remain $SRM$ codes invariant have not been studied yet. Motivated by this research gap, we have shifted our attention to tackling with this unsolved issue.
In Section \ref{sec : sym}, we provide some preliminaries about $SRM$ codes.
In Section \ref{sec : auto}, we investigate transformations that leave $SRM$ codes invariant. We determine all the linear transformations under which the $SRM$ codes are invariant for the parameters $n=2$ and $n=3$. Also, the set of these linear transformations forms a subgroup of the general linear group. 
We notice a relationship between the transformations that remain the $SRM$ codes invariant for $n=2$ and those for $n=3$. 
Therefore, we believe that we can form this group for a generic $n$, despite the challenge of determining all transformations. For any given $n$, we anticipate the group that leaves the $SRM$ codes invariant, akin to previous cases. This group involving affine transformations is a subgroup of the general linear group $GA(n,q)$. However, establishing whether this group forms a complete set leaving the $SRM$ codes invariant remains an open problem for the future. Therefore, in Section \ref{sec : Conc}, 
we state a conjecture on this problem.




\section{Symmetric Reed-Muller Codes}\label{sec : sym}
This section is devoted to some basic properties of Symmetric Reed-Muller ($SRM$) codes. These codes were first introduced by using bivariate polynomials over $\mathbb{F}_q$ in \cite{yan2020symmetric}. $SRM$ codes which exhibit specific symmetry properties within their codewords may be considered as subcodes of $GRM$ codes. The same authors generalized $SRM$ codes in another work \cite{yan2021local}. Before presenting the formal definition in \cite{yan2021local}, we recall some notions.

The set $E_q(n,r) \subseteq \mathbb{F}_q[x_1,x_2, \ldots,x_n]$ is defined by
\begin{align*}
    E_q(n,r) &:= 
    \Bigl\{ f(x_1,x_2,\ldots, x_n) = 
\sum_{\substack{0 \leq i_1 < i_2 < \ldots < i_n\leq q-1\\
i_1+i_2+\ldots + i_n\leq r}}a_{i_1i_2 \ldots i_n} det(x, i) \mid a_{i_1i_2 \ldots i_n} \in \mathbb{F}_q \Bigl\},
\end{align*}
where $x = (x_1, x_2, \ldots, x_n)$, $i = (i_1,i_2, \ldots, i_n)$ and
\begin{equation*}
     det (x,i) :=
    \begin{array}{c}
    \begin{vmatrix}
    {x_1}^{i_1} & {x_2}^{i_1} & \ldots & {x_n}^{i_1} \\
    {x_1}^{i_2} & {x_2}^{i_2} & \ldots & {x_n}^{i_2}  \\
    \vdots &\vdots & \ddots & \vdots \\
    {x_1}^{i_n} & {x_2}^{i_n} & \ldots & {x_n}^{i_n} \\
    \end{vmatrix}
    \end{array}.
\end{equation*}
Consider the set
$$
\Delta  = \{(a_1, a_2,\ldots,a_n) \in \mathbb{F}_q^n \mid a_j \neq a_i, 1\leq j<i\leq n \}$$
and an equivalence relation $\sim$ on $\Delta$ defined as follows:
\begin{center}
$c \sim d \iff \exists \sigma \in S_n$ such that $\sigma(c) = d$, for $c, d \in \mathbb{F}_q^n$. 
\end{center}

Then define the set $\Omega_q(n) := \Delta /\sim =\{[\alpha]\mid \alpha \in \Delta \}$, where $[\alpha ]$ denotes the equivalent class of $\alpha$. In order to get rid of dublications, the definition of $SRM$ codes is given by using this quotient set as follows: 
\begin{defi}\cite{yan2021local}
       Let $n$ and $r$ be positive integers, where $n$ denotes the number of variables. The $SRM$ code of degree $r$ over $\mathbb{F}_q$ is defined by
    $$ SRM_q[n, r] =\{ (f(\alpha))_{[\alpha]\in \Omega_q(n)} \mid f \in E_q(n, r)\}.$$
\end{defi}

\begin{rem}
According to the definition, when $q<r$, the possible $\{i_1,i_2,\dots,i_n\}$ sequence may not cover a partition of $r$. On the other hand, for the sequence $ \{i_1,i_2,\dots, i_n\}=\{0,1,2,\dots,n-1\}$ the smallest value of $r$ should be $\frac{n(n-1)}{2}$.
Thus, the definition of $SRM$ code is meaningful under the condition $q \geq r \geq \frac{n(n-1)}{2}$, where $q$ is chosen to be large enough.
\end{rem}

Note that under the condition mentioned in the remark above, when $n=1$, $SRM_q[1,r]$ codes are exactly generalized Reed-Solomon codes with degree parameter $r$.

When $n=2$ we have, 
    \begin{equation*}
        E_q(2,r) := \Bigg\{ \sum_{\substack{0 \leq i < j \leq q-1\\
i+j\leq r}}a_{ij}(x_1^ix_2^j- x_1^jx_2^i) \,\Bigg|\, a_{ij} \in \mathbb{F}_q \Bigg\} \subseteq \mathbb{F}_q[x_1,x_2].
    \end{equation*}
The evaluation of $f(x_1, x_2)$ at $(x_1, x_2) \in \mathbb{F}_q^2$ forms as the following matrix 
 \[ \begin{bmatrix}
        f(\alpha_0, \alpha_0) & f(\alpha_0, \alpha_1) & \ldots & f(\alpha_0, \alpha_{q-1}) \\
        f(\alpha_1, \alpha_0) & f(\alpha_1, \alpha_1) & \ldots & f(\alpha_1, \alpha_{q-1}) \\
        \vdots &\vdots & \ddots & \vdots \\
        f(\alpha_{q-1}, \alpha_0) & f(\alpha_{q-1}, \alpha_1) & \ldots & f(\alpha_{q-1}, \alpha_{q-1})
 \end{bmatrix}. \]
 Since $f(x_1,x_1)=0$ and $f(x_1,x_2)=-f(x_2,x_1)$, the above matrix is a skew-symmetric matrix. Thus it is entirely determined by the entries in the strictly upper triangular part. Therefore, the codeword of bivariate $SRM$, i.e., when $n=2$, codes are defined as the strictly upper triangular part of this matrix. Furthermore, $SRM_q[2,r]$ codes are of length $\frac{q(q-1)}{2}$.
 
 

\section{Linear Transformations Under Which $SRM$ Codes are Invariant}\label{sec : auto}

In this section, we will derive the invaraint groups of $SRM$ for $n=2$ and $n=3$ over the field $\mathbb{F}_p$, where $p$ is any prime number. For $n=2$ and $n=3$, we determine the exact set formed by transformations that leave $SRM$ codes invariant under a subgroup of the affine linear group. Different methods were employed to determine this set for the values of $n=2,3$, separately. The reason for using different methods is that one approach may not be highly suitable for the other. Our main theorems, which identify the group formed by transformations that leave the $SRM$ code invariant, and the necessary lemmas for these theorems are provided for each $n=2,3$.

\subsection{The case $n=2$}
Recall $SRM_q[2,r]$ is an evaluated code family whose evaluation polynomials come from the set $E_q(2,r)$,

    \begin{equation}
        E_q(2,r) = \left\{ \sum_{\substack{0 \leq i < j \leq q-1\\ i+j\leq r}}a_{ij}\begin{array}{c}
    \begin{vmatrix}
    {x_1}^{i} & {x_2}^{i} \\
    {x_1}^{j} & {x_2}^{j}\\
    \end{vmatrix}
    \end{array}\,\Bigg|\, a_{ij} \in \mathbb{F}_q \right\} .
    \end{equation}

    Let $f(x_1,x_2)=\begin{vmatrix}
    {x_1}^{i} & {x_2}^{i} \\
    {x_1}^{j} & {x_2}^{j}\\
    \end{vmatrix} \in E_q(2,r)$ and $A = \begin{bmatrix}
        a & b \\
        c & d
    \end{bmatrix}$. Under the  $T_f(A)$ transformation in (\ref{eqn:trans}), we get
    \begin{align*}
        T_f(A)&=
        \begin{vmatrix}
        (ax_1+bx_2)^{i} & (cx_1 +dx_2)^{i} \\
        (ax_1+bx_2)^{j} & (cx_1 +dx_2)^{j} 
        \end{vmatrix} \\
        &= (ax_1 + bx_2)^{i}\cdot(cx_1 +dx_2)^{j}-(ax_1+bx_2)^{j}\cdot(cx_1 +dx_2)^{i}.
    \end{align*}

    
We investigate the coefficients $a, b, c, d \in \mathbb{F}_q$ to find the transformations that keep the set $E_q(2,r)$ invariant. For this purpose, we need the following auxiliary lemmas:

\begin{lem}\label{alpha}

    Let $f \in E_q(2,r)$ and $ \alpha, \beta \in \mathbb{F}_q$ and $t$ be a positive integer. Then we have 
    \[(\alpha x_1^2 + \beta x_1x_2 + \alpha x_2^2)^t\cdot f \in E_q(2,r).\]  
\end{lem}
\begin{proof}
   Without loss of generality, we shall assume that 
   \[
   f =\begin{vmatrix}
       x_1^i & x_2^i\\
       x_1^j & x_2^j
   \end{vmatrix} = x_1^ix_2^j- x_1^jx_2^i,
   \] 
   where $0\leq i<j\leq (q-1)$ such $i+j\leq r$. We use the induction on $t$ to prove the claim. When $t = 1$, then 
    \begin{align*}
        (\alpha x_1^2 + \beta x_1x_2 +  \alpha x_2^2)\cdot f &= \alpha(x_1^2 +x_2^2)(x_1^ix_2^j- x_1^jx_2^i) + \beta (x_1x_2)(x_1^ix_2^j- x_1^jx_2^i) \\
        &= \alpha (f_1+f_2)+\beta f_3,
    \end{align*}
    where
    \[f_1=x_1^{i+2}x_2^j- x_1^jx_2^{i+2},\quad  f_2=x_1^{i}x_2^{j+2}- x_1^{j+2}x_2^i\quad \textnormal{and}\quad f_3=x_1^{i+1}x_2^{j+1}- x_1^{j+1}x_2^{i+1}.
    \]
    Thus $f_1, f_2, f_3 \in E_q(2,r)$, which implies that $(\alpha x_1^2 + \beta x_1x_2 +  \alpha x_2^2)\cdot f \in E_q(2,r)$.
    For the next step, suppose that the claim holds for $t>1$, i.e. we have
\begin{align}\label{ind : 0}
     (\alpha x_1^2 + \beta x_1x_2 + \alpha x_2^2)^{t}\cdot f \in E_q(2,r).
\end{align}
Then it is sufficient to show that the statement holds for $(t + 1)$. 
    \begin{align}\label{ind :1}
        (\alpha x_1^2 + \beta x_1x_2 + \alpha x_2^2)^{t + 1}\cdot f &= (\alpha x_1^2 + \beta x_1x_2 + \alpha x_2^2)((\alpha x_1^2 + \beta x_1x_2 + \alpha x_2^2)^{t}\cdot f).
    \end{align}
By the equation (\ref{ind : 0}), we have
\[
g=(\alpha x_1^2 + \beta x_1x_2 + \alpha x_2^2)^{t}\cdot f\in E_q(2,r).
\]
Since $g\in E_q(2,r)$, we may consider the equation (\ref{ind :1}) as follows:
\[
(\alpha x_1^2 + \beta x_1x_2 + \alpha x_2^2)\cdot g=(\alpha x_1^2 + \beta x_1x_2 + \alpha x_2^2)\cdot  \sum_{\substack{0 \leq i < j \leq q-1\\ i+j\leq r}}a_{ij} g_{ij},
\]
where $a_{ij}\in \mathbb{F}_q$ and
\[
g_{ij}=\begin{vmatrix}
       x_1^i & x_2^i\\
       x_1^j & x_2^j
   \end{vmatrix} = x_1^ix_2^j- x_1^jx_2^i.
\]
By the case of $t = 1$, we have all $(\alpha x_1^2 + \beta x_1x_2 + \alpha x_2^2)\cdot g_{ij} \in E_q(2,r)$. Thus $(\alpha x_1^2 + \beta x_1x_2 + \alpha x_2^2)\cdot g\in E_q(2,r)$. In the same manner, this is generalized for any $f\in E_q(2,r)$, which completes the proof.
\end{proof}

\begin{lem}\label{det}
Let $ a,b \in \mathbb{F}_q$ and $t$ be a positive integer, then
\[
\begin{vmatrix}
        1 & 1\\ 
        (ax_1+bx_2)^t & (bx_1+ax_2)^t
    \end{vmatrix} \in E_q(2,r).
\]
\end{lem}

\begin{proof}
When $t$ is odd, consider
    \begin{align*}
       \begin{vmatrix}
        1 & 1\\ 
        (ax_1+bx_2)^t & (bx_1+ax_2)^t
    \end{vmatrix} &= (bx_1+ax_2)^t - (ax_1+bx_2)^t \\
    &=\sum_{k=0}^{t}\binom{t}{k}\left( (bx_1+1)^{k}(ax_2)^{t-k}-(ax_1)^{k}(bx_2)^{t-k}\right)\\
    &=\sum_{k=0}^{t}\binom{t}{k}\left( b^{k}a^{t-k}-a^{k}b^{t-k}\right)x_1^kx_2^{t-k}\\
    &=\sum_{k=0}^{(t-1)/2}\binom{t}{k}\left( b^{k}a^{t-k}-a^{k}b^{t-k}\right)\left( x_1^kx_2^{t-k}-x_1^{t-k}x_2^{k}\right)\\
    &=\sum_{k=0}^{(t-1)/2}\binom{t}{k}\left( b^{k}a^{t-k}-a^{k}b^{t-k}\right)\begin{vmatrix}
        x_1^k & x_2^k\\
        x_1^{t-k} & x_2^{t-k}
    \end{vmatrix}.
    \end{align*}
Clearly, the corresponding determinant is an element of $E_q(2,r)$. 

When $t$ is even, similarly we have
\begin{align*}
       \begin{vmatrix}
        1 & 1\\ 
        (ax_1+bx_2)^t & (bx_1+ax_2)^t
    \end{vmatrix} &=\sum_{k=0}^{t}\binom{t}{k}\left( b^{k}a^{t-k}-a^{k}b^{t-k}\right)x_1^kx_2^{t-k}\\
    &=\sum_{k=0}^{t/2-1}\binom{t}{k}\left( b^{k}a^{t-k}-a^{k}b^{t-k}\right)\left( x_1^kx_2^{t-k}-x_1^{t-k}x_2^{k}\right).
\end{align*}
In the above equation when $k=t/2$, the coefficient of the term $x_1^{t/2}x_2^{t/2}$ is zero. Finally, as in the odd case, the corresponding determinant is an element of $E_q(2,r)$. Thus, the proof is completed.
\end{proof}

Afterward, we have the following proposition.

\begin{prop}\label{prop:n=2}
Let $a,b \in \mathbb{F}_q$ and $i$, $j$ be positive integers such that $i<j$, then
\[
\begin{vmatrix}
  (ax_1+bx_2)^i & (bx_1+ax_2)^i\\ 
        (ax_1+bx_2)^j & (bx_1+ax_2)^j 
\end{vmatrix} \in E_q(2,r).
\]
\end{prop}

\begin{proof}
The determinant can be written as: 
    \begin{align*}
        \begin{vmatrix}
        (ax_1+bx_2)^i & (bx_1+ax_2)^i\\ 
        (ax_1+bx_2)^j & (bx_1+ax_2)^j
    \end{vmatrix} &= (ax_1+bx_2)^i(bx_1+ax_2)^i\begin{vmatrix}
        1 & 1\\ 
        (ax_1+bx_2)^{j-i} & (bx_1+ax_2)^{j-i} 
    \end{vmatrix}\\[10pt]
    &= \left(abx_1^2+(a^2+b^2)x_1x_2+abx_2^2\right)^i\begin{vmatrix}
        1 & 1\\ 
        (ax_1+bx_2)^{t} & (bx_1+ax_2)^{t} 
    \end{vmatrix}
    \end{align*}
where $t=j-i$. By utilizing Lemmas \ref{det} and \ref{alpha}, the proof follows.

\end{proof}


The following theorem gives a necessary and sufficient condition for $SRM_q[2,r]$ to be invariant under which subgroup of $GL(2, q)$.



\begin{teo} Let $M$ be a set   defined as
\[
M =  \left\{ \begin{bmatrix}
      a & b \\
      b & a 
  \end{bmatrix} \mid a, b \in \mathbb{F}_q, \, a\neq \pm b
  \right\} \subset GL(2,q).
\]
The automorphism group of the $SRM_q[2,r]$, where $q\geq r>2$ except $q=r=3$, code family contains a subgroup isomorphic to $M$, i.e. $SRM_q[2,r]$ is invariant under the transformations come from $M$.
\end{teo}
\begin{proof}
    Let $A\in GL(2, q)$. Then take the transform $T$:
    \[ 
    T :
  \begin{bmatrix}
    x_1 \\
    x_2 \\
    
  \end{bmatrix}
   \mapsto
  A\begin{bmatrix}
    x_1 \\
    x_2 \\
    
  \end{bmatrix} =\begin{bmatrix}
      a & b \\
      c & d 
  \end{bmatrix}
  \begin{bmatrix}
    x_1 \\
    x_2 \\
    
  \end{bmatrix}. 
\]
In the set  $E_q(2,r)$, there exist unique polynomials of degree $1$ and $2$, which are
\[
f(x_1,x_2)=\begin{vmatrix}
    1 & 1 \\
    x_1 & x_2 
\end{vmatrix} = x_2 - x_1
\]
and
\[
g(x_1,x_2)=\begin{vmatrix}
    1 & 1 \\
    x_1^2 & x_2^2 
\end{vmatrix} = x_2^2-x_1^2,
\]
respectively. If $T_f(A)$ and $T_g(A)$ are elements of $E_q(2,r)$, it is easy to see that $T_f(A)$ and $T_g(A)$ must be scalar multiples of $f(x_1,x_2)$ and $g(x_1,x_2)$, respectively. In the light of this fact, we have
\begin{align*}
    T_f(A) = f(ax_1+bx_2,cx_1+dx_2)&=\begin{vmatrix}
    1 & 1 \\
    (ax_1+bx_2) & (cx_1+dx_2) 
\end{vmatrix} \\
&= (d-b)x_2-(a-c)x_1 
\end{align*}
and
\begin{align*}
    T_g(A) = g(ax_1+bx_2,cx_1+dx_2)&=\begin{vmatrix}
    1 & 1 \\
    (ax_1+bx_2)^2 & (cx_1+dx_2)^2 
\end{vmatrix}\\
&= (cx_1+dx_2)^2 - (ax_1+bx_2)^2\\
&= (d^2-b^2)x_2^2-(a^2-c^2)x_1^2+(2cd-2ab)x_1x_2.
\end{align*}
From the above, we obtain the following equations
\begin{align*}
a - c &= d - b,\\
(a^2-c^2)&=(b^2-d^2),\\ 
     (2cd- 2ab) &= 0.
\end{align*}
If we solve the equations above together with the fact $ad -bc \neq 0$, then we get $a = d$ and $b = c$. Combining this with Proposition \ref{prop:n=2}, $SRM_q[2,r]$ is invariant under the transformations that come from the set $M$, which completes the proof.
\end{proof}

Note that when $q=r=3$, the set $SRM_3[2,3]$ contains all vectors of length $3$ so that $Aut(SRM_3[2,3]) = S_{3}$.

In the following subsection, we will focus on the $SRM_q(3,r)$ in the same manner.
\subsection{The case $n=3$}

Recall that the code family $SRM_q[3,r]$ with the length $\sum_{i = 1}^{q-2}\frac{(i)(i+1)}{2}$ is an evaluated code family whose evaluation polynomials come from the set $E_q(3,r)$,

    \begin{equation}
        E_q(3,r) = \left\{ \sum_{\substack{0 \leq i < j < k \leq q-1\\ i+j + k\leq r}}a_{ijk}\begin{array}{c}
    \begin{vmatrix}
    {x_1}^{i} & {x_2}^{i} & {x_3}^{i}\\
    {x_1}^{j} & {x_2}^{j} & {x_3}^{j}\\
    {x_1}^{k} & {x_2}^{k} & {x_3}^{k}\\
    \end{vmatrix}
    \end{array}\,\Bigg|\, a_{ijk} \in \mathbb{F}_q \right\} .
    \end{equation}

Under the  $T_f(A)$ transformation in (\ref{eqn:trans}), where
\[
A= \begin{bmatrix}
    a & b & c \\
    d & e & f \\
    g & h & i
\end{bmatrix},
\]
we investigate the coefficients $a$, $b$, $c$, $d$, $e$, $f$, $g$, $h$, $i \in \mathbb{F}_q$ to find that keep the set $E_q(3,r)$ invariant. We require the auxiliary lemmas for this aim. 

The following lemma gives us a different interpretation of the set $E_q(3,r)$.

\begin{lem}\label{lemman=3}
    Let $f(x_1,x_2,x_3)\in\mathbb{F}_q[x_1,x_2,x_3]$ with a degree less than or equal to $r$, where $q$ is odd, such that for any $\pi\in S_3$,
\begin{equation}\label{lemma-n=3}
f(x_{\pi(1)},x_{\pi(2)},x_{\pi(3)}) = 
     \begin{cases}
        -f(x_1,x_2,x_3),  & \pi\,\,\text{is an odd permutation},\\
        f(x_1,x_2,x_3), & \pi \,\,\text{is an even permutation}. \\ 
     \end{cases}
\end{equation}
Then $f(x_1,x_2,x_3)\in E_q(3,r)$.
\end{lem}

\begin{proof}
Firstly, we may assume that we have a homogeneous nonzero multivariate polynomial $f(x_1,x_2,x_3)$ of degree $t\leq r$, with the property (\ref{lemma-n=3}).

Since $f$ is nonzero, we have a monomial term $A_0x_1^ix_2^jx_3^k$ in $f$. Without loss of generality, we may choose the powers $0 \leq i \leq j \leq k$, where $i, j, k$ are integers such that $i+j+k=t$. So we write $f$ as follows
\begin{equation}\label{lemma-initial}
f(x_1,x_2,x_3) = A_0x_1^ix_2^jx_3^k + g_0(x_1,x_2,x_3)
\end{equation}
where $A_0\in \mathbb{F}_q$ and $g_0(x_1,x_2,x_3)$ is a homogeneous polynomial of degree $t$ such that the coefficient of the $x_1^ix_2^jx_3^k$ in $g_0(x_1,x_2,x_3)$ is zero. 

Consider the case $i=j$. Then by the property (\ref{lemma-n=3}), we get
\begin{multline*}
f(x_2,x_1,x_3) = A_0x_2^ix_1^ix_3^k + g_0(x_2,x_1,x_3)\\
=-f(x_1,x_2,x_3)=-A_0x_1^ix_2^ix_3^k - g_0(x_1,x_2,x_3).
\end{multline*}
Equivalently, we have
\[
2A_0x_2^ix_1^ix_3^k + g_0(x_2,x_1,x_3)+ g_0(x_1,x_2,x_3)=0.
\]
Since the coefficient of the monomial $x_1^ix_2^ix_3^k$ in $g_0(x_1,x_2,x_3)$ is zero, $x_2^ix_1^ix_3^k$ must be a monomial term of $g_0(x_2,x_1,x_3)$, whose coefficient is $-2A_0$. This is the contradiction. Similarly, we get a contradiction for the cases $i=k$ and $j=k$. Thus there is not a monomial term $x_1^ix_2^jx_3^k$ such that at least two of $i,j$ and $k$ values are the same. 

By the cases mentioned above, we may assume that in the monomial term $x_1^ix_2^jx_3^k$, $i,j,k$ are distinct, i.e. $0 \leq i < j < k$. Consider the relation (\ref{lemma-initial}). When $\pi=(12)$, we have
\begin{multline*}
    f(x_2,x_1,x_3)  = A_0x_2^ix_1^jx_3^k + g_0(x_2,x_1,x_3)\\
    =-f(x_1,x_2,x_3)=-A_0x_1^ix_2^jx_3^k - g_0(x_1,x_2,x_3).
\end{multline*}
By the above equation, the monomial term $-A_0x_2^ix_1^jx_3^k$ must be in $g_0(x_1,x_2,x_3)$. Thus the relation (\ref{lemma-initial}) rewrite as
\[
f(x_1,x_2,x_3) = A_0x_1^ix_2^jx_3^k -A_0x_1^jx_2^ix_3^k+ g_1(x_1,x_2,x_3).
\]
Applying similar steps for the permutations $\pi=(13)$ and $\pi=(23)$, we get
\[
f(x_1,x_2,x_3) = A_0x_1^ix_2^jx_3^k -A_0x_1^jx_2^ix_3^k-A_0x_1^kx_2^jx_3^i-A_0x_1^ix_2^kx_3^j+ g_3(x_1,x_2,x_3)
\]
When applying the permutation $\pi=(123)$ to $f(x_1,x_2,x_3)$ in above relation, we get
\begin{multline*}
f(x_2,x_3,x_1) = A_0x_2^ix_3^jx_1^k -A_0x_1^jx_2^ix_3^k-A_0x_1^kx_2^jx_3^i-A_0x_1^ix_2^kx_3^j+ g_3(x_2,x_3,x_1)\\
=f(x_1,x_2,x_3) = A_0x_1^ix_2^jx_3^k -A_0x_1^jx_2^ix_3^k-A_0x_1^kx_2^jx_3^i-A_0x_1^ix_2^kx_3^j+ g_3(x_1,x_2,x_3),
\end{multline*}
which implies the monomial term $A_0x_2^ix_3^jx_1^k$ must be in $g_3(x_1,x_2,x_3)$. Thus
\[
f(x_1,x_2,x_3) = A_0x_1^ix_2^jx_3^k+A_0x_1^kx_2^ix_3^j -A_0x_1^jx_2^ix_3^k-A_0x_1^kx_2^jx_3^i-A_0x_1^ix_2^kx_3^j+ g_4(x_1,x_2,x_3).
\]
Finally, applying the permutation $\pi=(132)$, the polynomial $f(x_1,x_2,x_3)$ is of form
\begin{align*}
 &A_0x_1^ix_2^jx_3^k+A_0x_1^kx_2^ix_3^j+A_0x_1^jx_2^kx_3^i -A_0x_1^jx_2^ix_3^k-A_0x_1^kx_2^jx_3^i-A_0x_1^ix_2^kx_3^j+ g_5(x_1,x_2,x_3)\\
 =&A_0\begin{vmatrix}
     x_1^i &x_2^i &x_3^i\\
     x_1^j &x_2^j &x_3^j\\
     x_1^k &x_2^k &x_3^k
 \end{vmatrix}+ g_5(x_1,x_2,x_3),
\end{align*}
where $g_5(x_1,x_2,x_3)$ is a homogeneous polynomial of degree $t$ such that the coefficients of the monomials $x_{\pi(1)}^ix_{\pi(2)}^jx_{\pi(3)}^k$ for  any $\pi\in S_3$ are zero.

Thereafter, if we apply what we did for $f(x_1,x_2,x_3)$ to $g_5(x_1,x_2,x_3)$ by following the same steps for the other possible triple partition, $t = i_1 + j_1 + k_1$, we get
\[
f(x_1,x_2,x_3)=A_0\begin{vmatrix}
     x_1^i &x_2^i &x_3^i\\
     x_1^j &x_2^j &x_3^j\\
     x_1^k &x_2^k &x_3^k
 \end{vmatrix}+A_1\begin{vmatrix}
     x_1^{i_1} &x_2^{i_1} &x_3^{i_1}\\
     x_1^{j_1} &x_2^{j_1} &x_3^{j_1}\\
     x_1^{k_1} &x_2^{k_1} &x_3^{k_1}
 \end{vmatrix}+ g_6(x_1,x_2,x_3),
\]
where $A_0,A_1\in \mathbb{F}_q$ and $g_6(x_1,x_2,x_3)$ is a homogeneous polynomial of degree $t$.

Since the number of the triple partition of $t$ is finite, we may continue the above procedure until all possible partitions are over. Finally, the polynomial $f$ is of form
\[
f(x_1,x_2,x_3)=A_0\begin{vmatrix}
     x_1^i &x_2^i &x_3^i\\
     x_1^j &x_2^j &x_3^j\\
     x_1^k &x_2^k &x_3^k
 \end{vmatrix}+A_1\begin{vmatrix}
     x_1^{i_1} &x_2^{i_1} &x_3^{i_1}\\
     x_1^{j_1} &x_2^{j_1} &x_3^{j_1}\\
     x_1^{k_1} &x_2^{k_1} &x_3^{k_1}
 \end{vmatrix}+\cdots+A_d\begin{vmatrix}
     x_1^{i_d} &x_2^{i_d} &x_3^{i_d}\\
     x_1^{j_d} &x_2^{j_d} &x_3^{j_d}\\
     x_1^{k_d} &x_2^{k_d} &x_3^{k_d}
 \end{vmatrix},
\]
where $A_i$'s in $\mathbb{F}_q$. Thus $f(x_1,x_2,x_3) \in E_q(3,r)$ by the definition.

In general, for any polynomial $F(x_1,x_2,x_3)$ satisfies the condition (\ref{lemma-n=3}), we may write $F(x_1,x_2,x_3)$ as follows
\begin{align*}
     F(x_1,x_2,x_3) = f_3(x_1,x_2,x_3) + f_4(x_1,x_2,x_3) + \cdots +f_r(x_1,x_2,x_3),
\end{align*}
where $f_i$'s are homogeneous polynomials of degree $i$ for $i \in \{3, 4, \ldots,r\}$. From above for any $i \in \{3, 4, \ldots,r\}$, $f_i \in E_q(3,r)$. Thus $F(x_1,x_2,x_3)\in E_q(3,r)$, which completes the proof.\end{proof}

Let
\begin{align*}
    K=&\left\{ P * \begin{bmatrix}
        a & b & b \\
        b & a & b \\
        b & b & a
    \end{bmatrix},
    \Bigg|\, P \in \mathcal{P}_3, a,b \in \mathbb{F}_q, a\neq b,\, a\neq -2b\right\}\subset GL(3,q).
\end{align*}


\begin{lem}\label{lemma5}
    Let $A\in K$ and  $f(x_1, x_2, x_3) \in E_q(3,r) $, the function $g$ defined as $ g(x_1,x_2, x_3) = T_f(A)$ belongs the set $E_q(3,r)$.

\end{lem}
\begin{proof}
Firstly, assume that $A=\begin{bmatrix}
        b & a & a \\
        a & b & a \\
        a & a & b
    \end{bmatrix}$. Let $f(x_1,x_2,x_3) \in E_q(3,r)$. Under the transformation $T(A)$ in (\ref{def : T(A)}), we have variables $x_1 \mapsto bx_1 + ax_2 + ax_3$, $x_2  \mapsto ax_1 + bx_2 + ax_3$ and $x_3  \mapsto ax_1 + ax_2 + bx_3$. 
  
   Let $g=T_f(A)$:
   \begin{align*}
       g(x_1,x_2,x_3) &= T_f(A)=f( bx_1 + ax_2 + ax_3, ax_1 + bx_2 + ax_3,ax_1 + ax_2 + bx_3)
   \end{align*}
Let $\pi=(12)\in S_3$. 
\begin{align*}
g(x_{\pi(1)},x_{\pi(2)},x_{\pi(3)})=g(x_2,x_1,x_3)&=f( bx_2 + ax_1 + ax_3, ax_2 + bx_1 + ax_3,ax_2 + ax_1 + bx_3)\\
        &= -f( bx_1 + ax_2 + ax_3, ax_1 + bx_2 + ax_3,ax_1 + ax_2 + bx_3) \\
        &= -g(x_1,x_2,x_3),
\end{align*}
second line comes from the fact $f(x_1,x_2,x_3) \in E_q(3,r)$. Similarly, when $\pi=(13)$ or $\pi=(23)$, we get $g(x_{\pi(1)},x_{\pi(2)},x_{\pi(3)})=-g(x_1,x_2,x_3)$.

On the other hand, when $\pi=(123)\in S_3$, we have
\begin{align*}
g(x_{\pi(1)},x_{\pi(2)},x_{\pi(3)})=g(x_2,x_3,x_1)&=f(bx_2 + ax_3 + ax_1, ax_2 + bx_3 + ax_1,ax_2 + ax_3 + bx_1)\\
&=-f(ax_2 + bx_3 + ax_1,bx_2 + ax_3 + ax_1,ax_2 + ax_3 + bx_1)  \\
        &= f( bx_1 + ax_2 + ax_3, ax_1 + bx_2 + ax_3,ax_1 + ax_2 + bx_3) \\
        &= g(x_1,x_2,x_3).
\end{align*}
Similarly, when $\pi=(132)$ we obtain $g(x_{\pi(1)},x_{\pi(2)},x_{\pi(3)})=g(x_1,x_2,x_3)$.

Finally from above, we can characterize the multivariate polynomial $g(x_1,x_2,x_3)$ as follows for any $\pi\in S_3$:
   \[   
g(x_{\pi(1)},x_{\pi(2)},x_{\pi(3)})=
     \begin{cases}
       -g(x_1,x_2,x_3),  &  \text{$\pi$ is an odd permutation},\\
       g(x_1,x_2,x_3), & \text{$\pi$ is an even permutation}.  \\ 
     \end{cases}
\]
Thus by Lemma \ref{lemman=3}, $g=T_f(A)\in E_q(3,r)$. In the same manner, this approach can easily be applied to the other elements of the set $K$. So the proof is completed.
\end{proof}


\begin{teo}\label{main:n=3}
The automorphism group of the $SRM_q[3,r]$, where $q\geq r>3$, code family contains a subgroup isomorphic to $K$, i.e. $SRM_q[3,r]$ is invariant under the transformations come from $K$.
\end{teo}

\begin{proof}
By the definition of $E_q(3,r)$, the members of degrees $3$ and $4$ in $E_q(3,r)$ are the sets 
\[
P=\biggl\{ a_{012}\begin{vmatrix}
        1 & 1 & 1\\ 
       x_1 & x_2 & x_3 \\
       x_1^2 & x_2^2 & x_3^2
    \end{vmatrix} \Bigg|\, a_{012} \in \mathbb{F}_q\biggl\}
\]
and     
\[Q= \biggl\{ a_{013}\begin{vmatrix}
        1 & 1 & 1\\ 
        x_1 & x_2 & x_3 \\
        x_1^3 & x_2^3 & x_3^3
    \end{vmatrix} \Bigg|\, a_{013} \in \mathbb{F}_q\biggl\},
\]
respectively. Let $B\in\mathbb{F}_q^{3\times 3}$ be an invertible matrix such that  $B = \begin{bmatrix}
      a & b & c\\
      d & e & f \\
      g & h & i \\
  \end{bmatrix}$. It is clear that when $f\in P$ and $g\in Q$, the following conditions must hold
\begin{equation*}
T_f(B)\in P\qquad\textnormal{and}\qquad T_g(B)\in Q.
\end{equation*}

Thus, by the properties of the sets $P$ and $Q$, the coefficients of the terms 
\[
x_1^3,\,\, x_2^3,\,\, x_3^3,\,\, x_1x_2x_3,\,\, x_1^4,\,\, x_2^4,\,\, x_3^4,\,\, x_1x_2x_3^2, \,\, x_1x_2^2x_3,\,\, x_1^2x_2x_3,\,\, x_1^2x_2^2,\,\, x_1^2x_3^2,\,\, x_2^2x_3^2
\]
in the outputs, after applying the corresponding transformations $T_f(B)$ and $T_g(B)$ must be zero. Furthermore, in the same outputs, the sum of the coefficients of the terms in the pairs 
\[
(x_1x_2^2,x_1^2x_2),\,\, (x_1x_3^2,x_1^2x_3),\,\, (x_2x_3^2, x_2^2x_3),\,\, (x_1x_3^3,x_1^3x_3),\,\, (x_1x_2^3,x_1^3x_2),\,\, (x_2x_3^3, x_2^3x_3) 
\]
must individually be zero. So we have $19$ different equations over $\mathbb{F}_q$. When we solve these equations for the unknowns $a,b,c,d,e,f,g,h,i$ with the help of the computer algebra system SageMath, we get the solution set $S$ as 
\[
 \left\{ \begin{bmatrix}
        a & b & b \\
        b & a & b \\
        b & b & a
    \end{bmatrix},
    \begin{bmatrix}
        a & b & b \\
        b & a & b \\
        b & b & a
    \end{bmatrix},
    \begin{bmatrix}
        a & b & b \\
        b & b & a \\
        b & a & b
    \end{bmatrix},
    \begin{bmatrix}
        b & a & b \\
        a & b & b \\
        b & b & a
    \end{bmatrix},
    \begin{bmatrix}
        b & a & b \\
        b & b & a \\
        a & b & b
    \end{bmatrix},
    \begin{bmatrix}
        b & b & a \\
        b & a & b \\
        a & b & b
    \end{bmatrix},
    { \begin{bmatrix}
        b & b & a \\
        a & b & b \\
        b & b & a
    \end{bmatrix}}\bigg |\ a,b \in \mathbb{F}_q \right\}.
\]
For the invertibility for each element of the above set, $a,b\in\mathbb{F}_q$ satisfies the conditions $a\neq b,\,-2b$. Thus the solution set will be the set $K\subset GL(3,q)$. By combining this with Lemma \ref{lemma5}, the set $K$ is the maximal set in $GL(3,q)$ such that $E_q(3,r)$ is invariant under the transformations come from $K$.
\end{proof}

We give examples of $SRM_q[2,r]$ and $SRM_q[3,r]$ for some $q$, $r$ values, respectively.

\begin{ex}
Let $q = 5$, $n=2$, $r=4$ and $(i_1,i_2) \in \{(0,1), (0,2), (0,3), (0,4), (1,2), (1,3), \linebreak  (1,4), (2,3), (2,4), (3,4) \}$. 
For a matrix $\begin{bmatrix} 
a & c \\ 
b & d
\end{bmatrix}$
$\in \Biggl\{
\begin{bmatrix} 
1 & 0 \\ 
0 & 1
\end{bmatrix}, \begin{bmatrix} 
0 & 1 \\ 
1 & 0
\end{bmatrix}, \begin{bmatrix} 
0 & 2 \\ 
2 & 0 
\end{bmatrix}, \begin{bmatrix} 
2 & 0 \\ 
0 & 2
\end{bmatrix}, \begin{bmatrix} 
0 & 3 \\ 
3 & 0 
\end{bmatrix},  $

$\begin{bmatrix} 
3 & 0 \\ 
0 & 3
\end{bmatrix}, \begin{bmatrix} 
4 & 0 \\ 
0 & 4
\end{bmatrix}, \begin{bmatrix} 
0 & 4 \\ 
4 & 0
\end{bmatrix}, \begin{bmatrix} 
1 & 2 \\ 
2 & 1
\end{bmatrix}, \begin{bmatrix} 
2 & 1 \\ 
1 & 2
\end{bmatrix}, \begin{bmatrix} 
1 & 3 \\ 
3 & 1
\end{bmatrix}, \begin{bmatrix} 
3 & 1 \\ 
1 & 3
\end{bmatrix}, \begin{bmatrix} 
2 & 4 \\ 
4 & 2
\end{bmatrix},
\begin{bmatrix} 
4 & 2 \\ 
2 & 4
\end{bmatrix}, \begin{bmatrix} 
3 & 4 \\ 
4 & 3
\end{bmatrix}, \begin{bmatrix} 
4 & 3 \\ 
3 & 4
\end{bmatrix}\Biggr\}$ 
and $\alpha \in \mathbb{F}_5$, the following equation holds:
$$\alpha{ \Bigl[(x_1^{i_1}x_2^{i_2})-(x_1^{i_2}x_2^{i_1})\Bigr]} ={ (ax_1+bx_2)^{i_1}(cx_1 +dx_2)^{i_2} - (ax_1+bx_2)^{i_2}(cx_1 +dx_2)^{i_1}}$$

\end{ex}
\begin{ex}
    Let $q = 7$, $ n = 3$ and $r = 5$. Then $(i_1,i_2, i_3) \in \{(0,1,2), (0,1,3), (0,1,4), (0, 2, 3)\}$. 
    Let
    \begin{align*}
    g_1(x_1, x_2, x_3)&=x_1x_2^2+x_2x_3^2+x_1^2x_3-x_1^2x_2-x_2^2x_3-x_1x_3^2,\\
    g_2(x_1, x_2, x_3)&=x_1x_2^3+x_2x_3^3+x_1^3x_3-x_1^3x_2-x_2^3x_3-x_1x_3^3,\\
    g_3(x_1, x_2, x_3)&=x_1x_2^4+x_2x_3^4+x_1^4x_3-x_1^4x_2-x_2^4x_3-x_1x_3^4,\\
    g_4(x_1, x_2, x_3)&=x_1^2x_2^3+x_2^2x_3^3+x_1^3x_3^2-x_1^3x_2^2-x_2^3x_3^2-x_1^2x_3^3,
    \end{align*}
    and
    \[
    \mathcal{A}=\left\{ P \begin{bmatrix}
        a & b & b \\
        b & a & b \\
        b & b & a
    \end{bmatrix}
    \Bigg|\, P \in \mathcal{P}_3, a,b \in \mathbb{F}_7, a\neq b,\, a\neq -2b\right\}.
    \]
    Then
    \[
    E_7(3,5)=\{a_1g_1+a_2g_2+a_3g_3+a_4g_4\mid a_1,a_2,a_3,a_4\in \mathbb{F}_7\}.
    \]
    For the matrix $ K \in  \mathcal{A}$ and $g\in E_7(3,5)$, the polynomial $T_g(K)\in E_7(3,5)$ by Theorem \ref{main:n=3}. Since the $SRM_7[3,5]$ code is a type of polynomial evaluation codes, as in Example \ref{Abbe}, $SRM_7[3,5]$ is invariant under the corresponding transformations come from the $ \mathcal{A}$.
\end{ex}

Note that the determinants of matrices in solution set must be non-zero. For example, the solution set on $\mathbb{F}_5$ for $n=2$ does not include the element $\begin{bmatrix} 
1 & 4 \\ 
4 & 1
\end{bmatrix}$ because its determinant is zero. Additionally, as the parameter $n$ increases, the values of $q$ and $r$ should be adjusted accordingly.

\subsection{The general case}
Determining all transformations under 
 $GL(n,q)$ that leaves the $SRM$ code invariant for a general $n$ is quite challenging. For this, there needs to be a general method to identify such transformation. Nevertheless, we can predict the solution set that leaves the $SRM$ code invariant for a general $n$ and is a subgroup of the affine linear group. However, we have not established that this set may include all possible linear transformations under which $SRM$ codes are invariant. We leave the task of finding a generalized method for this problem for future work.

\section{Conclusion}\label{sec : Conc}
Our work aims at determining the set of affine-invariant transformations.
The linear automorphism groups of $SRM$ for $n = 2$ and $n = 3$ over the field $\mathbb{F}_p$, where $p$ is any prime number is proven in this paper.
For $n = 2$ and $n = 3$, we find that the exact set generated by transformations remaining $SRM$ codes invariant is a subgroup of the affine linear group. For different values of $n$, different techniques were used to determine this set. 
Therefore, we could not give a general proof for an arbitrary $n$, and leave it an open problem of the complete determination of the automorphism group $Aut(SRM)$ for any $n > 3$. We state our conjecture below.


\begin{conjecture}
Let $J_n$ be the $n \times n$ all one matrix, $I_n$ be the $n \times n$ identity matrix and $\mathcal{P}_n$ be the set of permutations of order $n$. Let $M$ be a subset of $GL(n,q)$ defined as
$M = \left\{ P((b-a)I_n + aJ_n) 
    \,|\, P \in \mathcal{P}_n, a,b \in \mathbb{F}_q, a\neq b, a\neq (1-n)b\right \}\subset GL(n,q)$. Then, the automorphism group of the $SRM_q[n,r]$ for $q>r>\frac{n(n-1)}{2}$ contains a subgroup isomorphic to $M$, i.e., $SRM_q[n,r]$ is invariant under the transformations in $M$.
\end{conjecture}

The proof of the invariance of $SRM_q[n,r]$ codes under the transformations, which come from the set $M$, may be similarly done to the proof of Lemma \ref{lemma5}. Notwithstanding, in order to show that the set $M$ is the complete set in this manner is quite challenging to follow the same techniques.
\subsection*{Acknowledgments:} The first author gratefully acknowledges the support she has received from The Scientific and Technological Research Council of Turkey (TUBITAK) with Grant No. 2211-A and The Council of Higher Education (YÖK) 100/2000 program. 
\bibliographystyle{plain}
\bibliography{sk_msthesis.bib}

\end{document}